\documentclass[letterpaper, 10 pt, conference]{ieeeconf}  % Comment this line out
\IEEEoverridecommandlockouts                              % This command is only
\usepackage{cite}
\usepackage{algpseudocode}
\usepackage{graphicx}
\usepackage{amsmath}
\usepackage{array}
\usepackage{caption}
\usepackage{subcaption}

\usepackage{comment}
\usepackage[font=scriptsize]{caption}
\usepackage{url}
\usepackage{txfonts}

\newtheorem{lemma}{Lemma}[section]
\newtheorem{proposition}{Proposition}[section]

\newtheorem{definition}{Definition}[section]

\usepackage{tikz}
\usetikzlibrary{arrows.meta, positioning}
\usetikzlibrary{calc}

\title{ \LARGE \bf Effect of Graph Gluing on Consensus in Networked \\ Multi-Agent Systems }

\author{Rohollah Moghadam,~\IEEEmembership{Senior Member,~IEEE} and Santosh Kandel
\thanks{Rohollah Moghadam is with the Department
of Electrical and Electronic Engineering, California State University-Sacramento,
5028 Riverside Hall, 6000 J Street, Sacramento, CA, 95819, USA {\tt\small moghadam@csus.edu}}% <-this % stops a space
\thanks{Santosh Kandel is with the Department
of Mathematics and Statistics, California State University-Sacramento,
226 Brighton Hall, 6000 J Street, Sacramento, CA, 95819, USA {\tt\small kandel@csus.edu}}% <-this % stops a space
}

\begin{document}

\maketitle
\thispagestyle{empty}
\pagestyle{empty}

\begin{abstract}
In this paper, the effects of graph gluing operations in networks of multi-agent systems and their impact on system performance are investigated. In many practical applications, multiple multi-agent subsystems must be interconnected through communication links to accomplish complex tasks, resulting in a larger communication network. Such interconnections modify the underlying graph topology and consequently affect the consensus behavior and convergence rate of the network. In particular, this paper examines both bridge gluing and interface gluing and analyzes how the number and structure of communication links between subsystems influence the Fiedler eigenvalue of the resulting graph. Since the Fiedler eigenvalue is directly related to the convergence rate of consensus dynamics, the proposed analysis establishes a clear relationship between interconnection strategies, algebraic connectivity, and system performance. The results provide theoretical insight into how different gluing mechanisms alter the spectral properties of the graph Laplacian and, in turn, the convergence characteristics of the networked multi-agent system. Simulation studies are presented to illustrate the theoretical findings and to validate the effectiveness of the proposed framework.
\end{abstract}
%\begin{IEEEkeywords}
%Internal Model Principle; Multi-Agent Systems; Cyber Physical Systems Attack; Consensus; Networked %Control Systems
%\end{IEEEkeywords}
\section{INTRODUCTION}
Distributed control of multi-agent cyber-physical systems has received considerable attention in recent years due to its wide range of applications. In these systems, multiple autonomous agents interact through a communication network and utilize locally available information from neighboring agents to cooperatively achieve a global objective \cite{Olfati2007,Fax2004,Jadbabaie2003}. A key advantage of distributed control is that each agent relies only on local information, which improves scalability, robustness, and flexibility in large-scale systems. Multi-agent systems have been applied in many areas, including coverage, where agents coordinate to monitor large environments, formation, where agents maintain desired geometric configurations while moving and surveillance, search and rescue, and target tracking, where teams of robots or drones collaborate to explore environments and collect information. These capabilities make distributed coordination strategies fundamental to many modern cyber-physical systems. 

As the range of applications of multi-agent systems continues to grow \cite{Chai2025}, multiple multi-agent subsystems are often interconnected to accomplish more complex tasks and improve overall system performance. For instance, in search-and-rescue operations \cite{Cao2024,Ge2026}, a swarm of aerial drones may form one multi-agent subsystem responsible for rapid area exploration, while a team of unmanned ground vehicles (UGVs) forms another subsystem responsible for ground-level inspection and target tracking. By connecting these subsystems through communication links, the overall system can coordinate information sharing and decision-making, enabling faster detection, improved analysis, and more efficient tracking of targets.

Although multi-agent systems have been extensively studied in the literature, most existing works focus on coordination and control within a single network of agents. To the best of our knowledge, limited attention has been given to the problem of interconnecting multiple multi-agent subsystems to form a larger network of networks. In many emerging applications, groups of multi-agent systems must cooperate and exchange information to accomplish complex missions. However, connecting multiple multi-agent systems fundamentally changes the underlying communication topology, resulting in a larger and more complex graph structure. This interconnection affects the spectral properties of the overall graph and consequently influences the consensus behavior and convergence characteristics of the network, since all agents contribute to the dynamics of the combined system. Therefore, understanding how the interconnection of multi-agent subsystems impacts consensus performance is essential for the design and operation of large-scale networked multi-agent systems.

% The idea of ``cutting” and   ``gluing”, which is similar to the divide and conquer strategy, is widely used in Mathematics and Statistics \cite{reshetikhin2010lectures, Kandel2021, ng2001}. Recently, some literature have explored applications of the cutting and gluing approach to study problems in Combinatorics \cite{contreras2020, contreras2024}. For example, in \cite{contreras2024}, the effect of gluing on so-called Feynman diagrams is explored. In \cite{contreras2020}, the effect of gluing on the characteristics polynomial of the graph Laplacian is studied. Recently, in \cite{brar2025gluing}, a gluing formula for the pseudo-determinant of the graph Laplacian is derived and its application on counting spanning trees on undirected graphs is explored.  

The idea of “cutting” and “gluing,” similar to the divide and conquer paradigm, is widely used in mathematics and statistics \cite{reshetikhin2010lectures, Kandel2021, ng2001}. Some recent works have explored applications of this approach to problems in combinatorics \cite{contreras2020, contreras2024}. For instance, \cite{contreras2024} examines the effect of gluing on so-called Feynman diagrams associated with combinatorial quantum field theory, while \cite{contreras2020} studies its impact on the characteristic polynomial of the graph Laplacian. More recently, \cite{brar2025gluing} derives a gluing formula for the pseudo-determinant of the graph Laplacian and applies it to the enumeration of spanning trees in undirected graphs.

In this paper, we investigate the effect of graph gluing, where multiple multi-agent system graphs are connected through communication links, on the convergence and consensus time of the resulting network. We particularly analyze the effect of graph gluing on the Fiedler eigenvalue of the overall graph that is directly related to the convergence time. The results provide insight into how the interconnection structure influences system performance, which is critical for the reliable operation of large-scale networked multi-agent systems and mission-critical applications.

 %%%%%%%%%%%%%%%%%%%%%%%%%%%%%%%% Section 2 %%%%%%%%%%%%%%%%%%%%%%%%%%%%%%%%%%%%%%%%%%%%%%%%%%%%%%%%%
\section{PRELIMINARY}
In this section, some background of graph theory and some definitions are provided.

\subsection{Basic Definitions and Notations from Graph Theory} 

An undirected graph $\mathcal{G}$ with $N$ nodes, consists of a pair $\left( {\mathcal{V}(\mathcal{G}),\mathcal{E}(\mathcal{G})} \right)$ in which $\mathcal{V}(\mathcal{G}){\text{ =  }}\{ {v_1}, \cdots ,{v_N}\} $ is a set of nodes and  $\mathcal{E}(\mathcal{G})\subseteq \mathcal{V}(\mathcal{G}) \times \mathcal{V}(\mathcal{G})$ is a set of edges. The adjacency matrix is defined as $A = \left[ {{a_{ij}}} \right]$, with ${a_{ij}}=1$ if $({v_j},{v_i}) \in \mathcal{E}(\mathcal{G})$, and ${a_{ij}} = 0$, otherwise. The nodes $i$ and $j$ are adjacent, if there is an edge between them. A path from node $i$ to node $j$ is a sequence of distinct nodes starting with node $i$ and ending with node $j$ while consecutive nodes are adjacent. The graph $\mathcal{G}$ is called connected, if there exists a path between any two nodes of the graph $\mathcal{G}$. The degree matrix $D \in {\mathbb{R}^{N \times N}}$ is a diagonal matrix with diagonal elements ${d_i}=\#\{j:(j,i) \in \mathcal{E}(\mathcal{G})\}$. The graph Laplacian matrix of $\mathcal{G}$ is defined as ${L}=D-A$ and it is a symmetric matrix for undirected graphs. For an undirected and connected graph $\mathcal{G}$ with $N$ vertices, $0$ is a simple eigenvalue and all eigenvalues can be arranged in an increasing order as $0=\lambda_1<\lambda_2<\dots<\lambda_N$. The eigenvector corresponding to zero eigenvalue is ${{\mathbf{1}}_N}$. 

\section{PROBLEM FORMULATION}
Consider a network of $N$ agents described by a communication graph 
$\mathcal{G} = (\mathcal{V}, \mathcal{E})$. Here, it is assumed that each agent is modeled by first-order dynamics
\begin{equation}
\dot{x}_i(t) = u_i(t), \quad i = 1,\dots,N,
\end{equation}
where $x_i(t) \in \mathbb{R}$ is the state of agent $i$. The control input is 
designed using a distributed consensus protocol
\begin{equation}
u_i(t) = \sum_{j \in \mathcal{N}_i} a_{ij} \big(x_j(t) - x_i(t)\big),
\end{equation}
which leads to the global dynamics
\begin{equation}
\dot{x}(t) = -L x(t).
\end{equation}

The objective is to achieve consensus, i.e.,
\begin{equation}
\lim_{t \to \infty} x_i(t) = x_j(t), \quad \forall i,j \in \mathcal{V}.
\end{equation}
It is well known \cite{lewis2013cooperative} that the convergence rate is governed by the Fiedler eigenvalue $\lambda_2(L)$, with the consensus time constant satisfying $\tau \propto 1/\lambda_2(L)$.

In this paper, we consider a multi-agent system formed by interconnecting 
subgraphs through graph gluing operations (e.g., bridge and interface gluing). 
Let $\mathcal{G}_1$ and $\mathcal{G}_2$ be two connected graphs and 
$\mathcal{G}$ denote the resulting graph after gluing. The main objective is to analyze how the interconnection structure affects 
the Laplacian spectrum, particularly the Fiedler eigenvalue $\lambda_2(L)$, 
and consequently the convergence behavior of the overall multi-agent system.
\medskip

\noindent
\textbf{Remark 1.} It is worth noting that the agents may have general linear or nonlinear dynamics. The analysis presented in this paper primarily focuses on the influence of the communication graph topology. For simplicity and clarity of exposition, we adopt single-integrator dynamics, however, this assumption is not restrictive, as the key results depend on graph-theoretic properties and can be extended to more general agent dynamics.

\section{GRAPH GLUING AND THEIR IMPACT ON ALGEBRAIC CONNECTIVITY}
In this section, we consider the gluing operations between graphs and study their effect on the algebraic connectivity of graphs. Two distinct scenarios as bridge gluing and interface gluing are explored. 

Let $\mathcal{G}$ be an undirected graph. In \cite{fiedler1973}, Fiedler proved that $\mathcal{G}$ is connected if and only if $\lambda_2(\mathcal{G})>0$. Fiedler used the eigenvalue $\lambda_2(\mathcal{G})$ to define the algebraic connectivity of the graph.  There is a large amount of literature that are devoted to find bounds for $\lambda_2(\mathcal{G})$, please refer to \cite{wu2007} and references there in for more details.

\subsection{Gluing Operations on Graphs}
Let us introduce two different gluing operations.

\begin{definition}[Bridge Gluing \cite{contreras2020}] Let $\mathcal{G}_1$ and $\mathcal{G}_2$ be two graphs; $v_1^{(1)}, \ldots,  v_k^{(1)}$ be vertices in $\mathcal{G}_1$; and $v_1^{(2)}, \ldots,  v_k^{(2)}$ be vertices in $\mathcal{G}_2$. A bridge $B$ is a graph with vertices 
$\left\{v_1^{(1)}, \ldots,  v_k^{(1)}, \,v_1^{(2)}, \ldots,  v_k^{(2)}\right\}$
 and  $\mathcal{E}(B)= \left\{(v_i^{(1)}, v_i^{(2)})\,\big| 1\leq i\leq k\right\}$.  The bridge gluing of $\mathcal{G}_1$ and $\mathcal{G}_2$ along the bridge $B$ is a graph denoted by $\mathcal{G}_1\sqcup_{B}\mathcal{G}_2$ and defined by 
 \begin{align*}
& \mathcal{V}(\mathcal{G}_1\sqcup_B\mathcal{G}_2) = \mathcal{V}(\mathcal{G}_1)\cup \mathcal{V}(\mathcal{G}_2) \\
& \mathcal{E}(\mathcal{G}_1\sqcup_B\mathcal{G}_2) = \mathcal{E}(\mathcal{G}_1) \cup \mathcal{E}(B) \cup \mathcal{E}(\mathcal{G}_2)
 \end{align*}
 \end{definition}

\medskip 
\begin{definition}[Interface Gluing \cite{contreras2020}]
Let $\mathcal{G}_1$ and $\mathcal{G}_2$ be two graphs and $Y$ be a {full subgraph} of $\mathcal{G}_1$ and $\mathcal{G}_2$. Here, a full subgraph means if there is an edge joining two vertices in the graph, then it is also an edge in the subgraph. Then, we can construct a new graph $\mathcal{G} = \mathcal{G}_1\cup_{Y} \mathcal{G}_2$ by identifying the vertices and edges present in $Y$: $\mathcal{V}(\mathcal{G}_1 \cup_{Y} \mathcal{G}_2) = \mathcal{V}(\mathcal{G}_1)\sqcup \mathcal{V}(\mathcal{G}_2)/\sim $ and $\mathcal{E}(\mathcal{G}_1 \cup_{Y} \mathcal{G}_2) = \mathcal{E}(\mathcal{G}_1) \sqcup (\mathcal{E}(\mathcal{G}_2)/\sim $. Here, $v^{(1)} \in \mathcal{V}(\mathcal{G}_1)$ is identified with $v^{(2)} \in \mathcal{V}(\mathcal{G}_2)$, i.e, $v^{(1)}\sim v^{(2)}$ if they represent the same vertex in $Y$.  Similarly, $e^{(1)} \in \mathcal{E}(\mathcal{G}_1)$ is identified with $e^{(2)} \in \mathcal{E}(\mathcal{G}_2)$, i.e., $e^{(1)}\sim e^{(2)}$ if they represent the same edge in $Y$.  In this case,
 \begin{align*}
& \mathcal{V}(\mathcal{G}_1\cup_Y\mathcal{G}_2) = \left(\mathcal{V}(\mathcal{G}_1)\setminus \mathcal{V}(Y)\right) \cup  \left(\mathcal{V}(\mathcal{G}_2)\setminus \mathcal{V}(Y)\right) \cup \mathcal{V}(Y)\\
& \mathcal{E}(\mathcal{G}_1\cup_Y\mathcal{G}_2) = \left(\mathcal{E}(\mathcal{G}_1)\setminus \mathcal{E}(Y)\right) \cup \left(\mathcal{E}(\mathcal{G}_2)\setminus \mathcal{E}(Y)\right)\cup \mathcal{E}(Y)
 \end{align*}
 \end{definition}

\subsection{Bridge Gluing of Graphs and Bounds on Fiedler's Eigenvalues}
In this subsection, bounds on $\lambda_2(G)$ are derived for $\mathcal{G} = \mathcal{G}_1\sqcup_B \mathcal{G}_2$. In particular, an upper bound is obtained using the notion of the graph cut. Let $S_1$ and $S_2$ be subset of $\mathcal{V}(\mathcal{G})$. Then, the cut of $S_1$ and $S_2$ is defined as 
\begin{equation*}
\mathrm{cut}(S_1, S_2) = \sum_{u \in S_1, v\in S_2}a_{uv}
\end{equation*} 
where $a_{uv}$ is the component of the adjacency matrix associated to the edge $\{u,\,v\}$ joining $u$ to $v$.  For an unweighted graph,  $\mathrm{cut}(S_1, S_2)$ is simply the total number of edges joining vertices in $S_1$  to vertices in $S_2$.  The following lemma gives an upper bound on $\lambda_2(\mathcal{G})$ of  an undirected graph. 

\begin{lemma}[\cite{wu2007}]\label{lemma:wu2007} Let $S_1$ and $S_2$ be two nontrivial disjoint subsets of $\mathcal{V}(\mathcal{G})$ of a graph $\mathcal{G}$. Let $S_i^c = \mathcal{V}(\mathcal{G})\setminus S_i, \, i=1,2$.  Then,
\begin{equation}
\lambda_2(\mathcal{G}) \leq \dfrac{\mathrm{cut}(S_1, S_1^c)}{|S_1|} + \dfrac{\mathrm{cut}(S_2, S_2^c)}{|S_2|} 
\end{equation}
\end{lemma}
\medskip

This lemma immediately provides an upper bound for $\lambda_2(\mathcal{G})$ whenever,  $\mathcal{G}= \mathcal{G}_1 \sqcup_B \mathcal{G}_2$.

\begin{proposition}\label{Prop:Bridge Gluing} Let $\mathcal{G} = \mathcal{G}_1 \sqcup_B \mathcal{G}_2$.  Assume that $B$ has k-edges.  Then,

\begin{equation}\label{gluing}
\lambda_2(\mathcal{G}) \leq \dfrac{k}{|\mathcal{V}(\mathcal{G}_1)|} + \dfrac{k}{|\mathcal{V}(\mathcal{G}_2)|}  
\end{equation}

\begin{proof}
Using $S_1 = \mathcal{V}(\mathcal{G}_1)$ and $S_2 = \mathcal{V}(\mathcal{G}_2)$ in Lemma \ref{lemma:wu2007} and observing $\mathrm{cut}(S_1, S_1^c)= \mathrm{cut}(S_2, S_2^c) = k$,  we get
\[
\lambda_2(\mathcal{G}) \leq  \dfrac{k}{|\mathcal{V}(\mathcal{G}_1)|} + \dfrac{k}{\mathcal{V}(\mathcal{G}_2)|}
\]
\end{proof}
\end{proposition}
\medskip

From Proposition \ref{Prop:Bridge Gluing}, it follows that  $\lambda_2(\mathcal{G})$ becomes smaller for smaller values of $k$, which is compatible with the interpretation of $\lambda_2(\mathcal{G})$ as it contains information on the connectivity of the graph. Geometrically, for smaller values of $k$ the graph is ``less`` connected and the Fiedler's eigenvalue is expected to be smaller. In addition, from Proposition \ref{Prop:Bridge Gluing}, it follows that if the number of vertices on $\mathcal{G}_1$ or $\mathcal{G}_2$ are increased, then, $\lambda_2(\mathcal{G})$ decreases.

It is known that the consensus time constant is inversely proportional to the Fiedler eigenvalue, i.e., $\tau=1/\lambda_2$ \cite{lewis2013cooperative}. From \eqref{gluing}, as the number of bridge edges $k$ increases, the Fiedler eigenvalue increases, resulting in a faster convergence rate.
\medskip

\subsection{Interface Gluing of Graphs and Bounds of Fiedler's Eigenvalue}
In this subsection, the effect of interface gluing on the Fiedler's eigenvalue is analyzed. Let $\mathcal{G} = \mathcal{G}_1 \cup_Y \mathcal{G}_2$. For each $ i \in \{1,2\}$, the Laplacian $L(\mathcal{G}_i)$ is expressed in block matrix form with respect to the interface set $Y$:
\[L(\mathcal{G}_i)
 = \begin{bmatrix}
      A_i& B_i \\[1mm]
         B_i^{T}& D_i
 \end{bmatrix},
\] where the matrices $D_i$ are indexed by vertices of $Y$, and the matrices $A_i$ are indexed by the vertices that are not in $Y$.  It can be shown that $A_1$ and $A_2$ are symmetric positive definite matrices \cite{contreras2024}. Let $\lambda_1(A_1)$ and $\lambda_1(A_2)$ be the smallest eigenvalues of $A_1$ and $A_2$ respectively. We can use these eigenvalues to find an upper bound for $\lambda_2(\mathcal{G})$. This is the content of the following proposition. 

\begin{proposition}\label{Prop:Interface Gluing} The following holds:
\[
\lambda_2(\mathcal{G}) \leq {\max}\{\lambda_1(A_1), \lambda_1(A_2)\}
\]

\begin{proof} The proof of this proposition is based on the application of Poincare Min--Max Principle for eigenvalues of a symmetric matrix. This proof is similar to that of \cite{Benson2015} on a  Riemannian manifold and \cite{balti2017} for the so-called special Laplacian on a directed graph. 

By Poincare Min-Max Principle for eigenvalue, we know that 
\begin{equation}\label{ineq: Poincare Min Max}
\lambda_2(\mathcal{G})  =\inf\limits_{\{V|\,\mathrm{dim}(V)\, =\,2\}} \sup\limits_{ v \in V\setminus \{0\}} \dfrac{v^T L(\mathcal{G}) v}{v^Tv}
\end{equation}

Let $v_1$ be eigenvector of $A_1$ associated to $\lambda_1(A_1)$  and $v_1$ be eigenvector of $A_2$ associated to $\lambda_1(A_2)$.   Define $\widetilde{v_1} = \begin{bmatrix}v_1 \\ \boldsymbol{0} \end{bmatrix}$ and $\widetilde{v_2} = \begin{bmatrix}\boldsymbol{0} \\  v_2 \end{bmatrix}$ so that $\widetilde{v_1}, \widetilde{v_2} $ are in $\mathbb{R}^{\mathcal{V}(\mathcal{G})}$. Note that $\widetilde{v_1}, \widetilde{v_2} $ are linearly independent by construction.  Let $V = \mathrm{span}\{\widetilde{v_1}, \widetilde{v_2}\}$.  Furthermore,  for any $v\in V$,  we can find unique $u_1  = c_1 v_1$ and $u_2  = c_2 v_2$ such that $v^T L(\mathcal{G}) v = u_1^TA_1 u_1 + u_2^TA_2 u_2$  and $v^Tv = u_1^T u_1 + u_2^T u_2$.  Moreover,
\begin{align*}
v^T L(\mathcal{G}) v &= u_1^TA_1 u_1 + u_2^TA_2 u_2\\
 & = \lambda_1(A_1)u_1^Tu_1 + \lambda_1(A_2)u_2^Tu_2 \\
 & \leq {\max}\{\lambda_1(A_1), \lambda_1(A_2)\} (u_1^T u_1 + u_2^T u_2)\\
 & = {\max}\{\lambda_1(A_1), \lambda_1(A_2)\}  v^Tv
\end{align*}
Now, from (\ref{ineq: Poincare Min Max}) it follows that 
\[
\lambda_2(\mathcal{G}) \leq \sup\limits_{ v \in V\setminus \{0\}} \dfrac{v^T L(\mathcal{G}) v}{v^Tv}
  \leq  {\max}\{\lambda_1(A_1), \lambda_1(A_2)\}. 
\]
\end{proof} 
\end{proposition}

The inequality in the Proposition \ref{Prop:Interface Gluing} is optimal in the sense that equality can be achieved in some cases. Let $\mathcal{G}_1$ be the path graph with two vertices and $\mathcal{G}_2$ is also a path graph with two vertices.  If $\mathcal{G}_1$ and $\mathcal{G}_2$ are glued along a common vertex,  then,  $\lambda_1(A_1) = 1$,  $\lambda_1(A_2) =1$ and $\lambda_2(\mathcal{G}) = 1$. 

\smallskip

\noindent
\textbf{Remark 2.} The concept of graph gluing is particularly relevant to the cybersecurity of multi-agent systems \cite{Rohollah2017CDC,Zhang2021_Survey}, because when a compromised multi-agent system is connected to another multi-agent network, malicious effects propagates to the previously intact network and continue to influence its performance even if the two networks are later disconnected. Moreover, an attacker can target the communication links between the two graphs, effectively reducing the number of interconnecting edges $k$ in \eqref{gluing}. This leads to a decrease in the Fiedler eigenvalue and, consequently, a slower convergence of the overall network. 
\medskip

\noindent
\textbf{Remark 3.} The proposed framework is also applicable to the decomposition of large-scale graphs into smaller components. Specifically, it enables the identification of critical edges or nodes whose removal leads to graph disconnection, and provides insight into how the resulting subgraphs evolve toward potentially different consensus states for the smaller graphs.
\medskip

\noindent
\textbf{Remark 4.} The proposed framework is particularly relevant to multi-robot collaboration, where heterogeneous teams of unmanned aerial vehicles (UAVs) and unmanned ground vehicles (UGVs) can dynamically interconnect, i.e. via wireless communication, to form larger networked multi-agent systems. These interconnections enhance coverage, information sharing, and robustness, thereby improving performance in applications such as search and rescue and target tracking. Consequently, analyzing how such graph combinations affect connectivity and consensus properties is crucial for the design of reliable cooperative control strategies.
\medskip

\noindent
\textbf{Remark 5.} The results in Proposition \ref{Prop:Bridge Gluing} and Proposition \ref{Prop:Interface Gluing} can be extended to directed graph. In this case, since the Laplacian matrix is generally nonsymmetric, it is more appropriate to analyze either the real part of the Fiedler eigenvalue or the algebraic connectivity of the directed graph \cite{wu2007}.
\medskip

\noindent
\textbf{Remark 6.} Interface gluing offers several practical advantages in networked multi-agent systems. Unlike bridge gluing, it does not require additional communication links between subsystems, as the interconnection is achieved through shared interface nodes. This reduces communication overhead and simplifies implementation. Furthermore, the shared interface nodes can provide inherent redundancy in the event of failures or adversarial disruptions (e.g., cyber-attacks) affecting certain agents, these nodes can be replaced or supported by backup agents within the interface set. This enhances the resilience and robustness of the overall network.

\section{SIMULATION RESULTS}
In this section, simulation examples are provided to verify the effectiveness of the proposed gluing approach in networked multi-agent systems. 

\subsection{Example 1 (Bridge gluing in multi-agent systems with small graphs)}
Consider two multi-agent systems with single integrator dynamics $\dot{x}_i = u_i$ and the distributed controller $u_i=\sum (x_i-x_j)$ are communicating over a graph topology depicted in Fig. \ref{fig:two-graphs}. The initial conditions of agents in Graph 1 are $x_1(0)=5,x_2(0)=-3,x_3(0)=4$ and for Graph 2 are $x_4(0)=-5,x_5(0)=2,x_6(0)=-1$.
\begin{figure}[h]
\centering
\begin{tikzpicture}[
    every node/.style = {
        circle, draw,
        line width=2pt,
        minimum size = 0.65cm,
        font = \bfseries
    }
]

%% --- Red Triangle (1,2,3) ---
\node[draw=red, text=red] (1) at (0, 2)    {1};
\node[draw=red, text=red] (2) at (-1, 0)   {2};
\node[draw=red, text=red] (3) at (1, 0)    {3};

\draw[red, line width=2pt] (1) -- (2);
\draw[red, line width=2pt] (2) -- (3);
%\draw[red, line width=2pt] (3) -- (1);

%% --- Green Triangle (4,5,6) ---
\node[draw=green!60!black, text=green!60!black] (4) at (3.5, 2)  {4};
\node[draw=green!60!black, text=green!60!black] (5) at (2.5, 0)  {5};
\node[draw=green!60!black, text=green!60!black] (6) at (4.5, 0)  {6};

\draw[green!60!black, line width=2pt] (4) -- (5);
\draw[green!60!black, line width=2pt] (5) -- (6);
%\draw[green!60!black, line width=2pt] (6) -- (4);

\end{tikzpicture}
\caption{Communication topology of two individual graphs.}
\label{fig:two-graphs}
\end{figure}
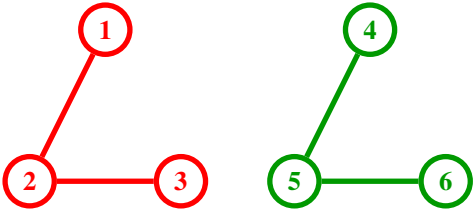

The graph Laplacian matrix associated to each graph is
\begin{equation}
    L_1 = \left[ {\begin{array}{*{20}{c}}
  1&{ - 1}&0 \\ 
  { - 1}&2&{ - 1}\\ 
  0&{ - 1}&1
\end{array}} \right], \quad L_2 = \left[ {\begin{array}{*{20}{c}}
  1&{ - 1}&0 \\ 
  { - 1}&2&{ - 1}\\ 
  0&{ - 1}&1
\end{array}} \right]
\end{equation}

\begin{figure}[h]
\centering
\begin{tikzpicture}[
    every node/.style = {
        circle, draw,
        line width=2pt,
        minimum size = 0.65cm,
        font = \bfseries
    }
]

%% --- Red Triangle (1,2,3) ---
\node[draw=red, text=red] (1) at (0, 2)    {1};
\node[draw=red, text=red] (2) at (-1, 0)   {2};
\node[draw=red, text=red] (3) at (1, 0)    {3};

\draw[red, line width=2pt] (1) -- (2);
\draw[red, line width=2pt] (2) -- (3);
%\draw[red, line width=2pt] (3) -- (1);

%% --- Green Triangle (4,5,6) ---
\node[draw=green!60!black, text=green!60!black] (4) at (3.5, 2)  {4};
\node[draw=green!60!black, text=green!60!black] (5) at (2.5, 0)  {5};
\node[draw=green!60!black, text=green!60!black] (6) at (4.5, 0)  {6};

\draw[green!60!black, line width=2pt] (4) -- (5);
\draw[green!60!black, line width=2pt] (5) -- (6);
%\draw[green!60!black, line width=2pt] (6) -- (4);

%% --- Red connecting line from 1 to 4 ---
\draw[blue, line width=2pt] (1) -- (4);

\end{tikzpicture}
\caption{The graph gluing with one bridge.}
\label{fig:one-bridge}
\end{figure}
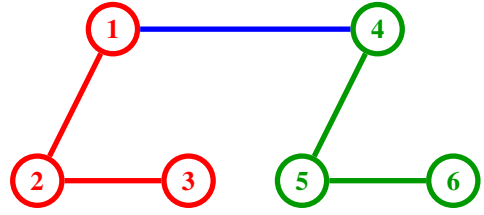

The second eigenvalue of each graph is $\lambda_2=1$. Now gluing Graph 1 and Graph 2 with a single bridge, i.e. $k=1$ in \eqref{gluing}, between node 1 and 4 as shown in Fig. \ref{fig:one-bridge} results in the following graph Laplacian matrix 
\begin{equation}
    L_{12}^1 = \left[ {\begin{array}{*{20}{c}}
  2&{ - 1}&0 &-1 &0 &0 \\ 
  { - 1}&2&{ - 1}&0 & 0&0\\ 
  0&{ - 1}&1 & 0& 0&0\\
  -1&0&0&2 &-1 &0\\
  0&0&0 & -1&1 &0\\
  0&0&0 &0 & -1&1\\
\end{array}} \right]
\end{equation}

\begin{figure}[!ht]
\begin{center}
\includegraphics[scale=0.45]{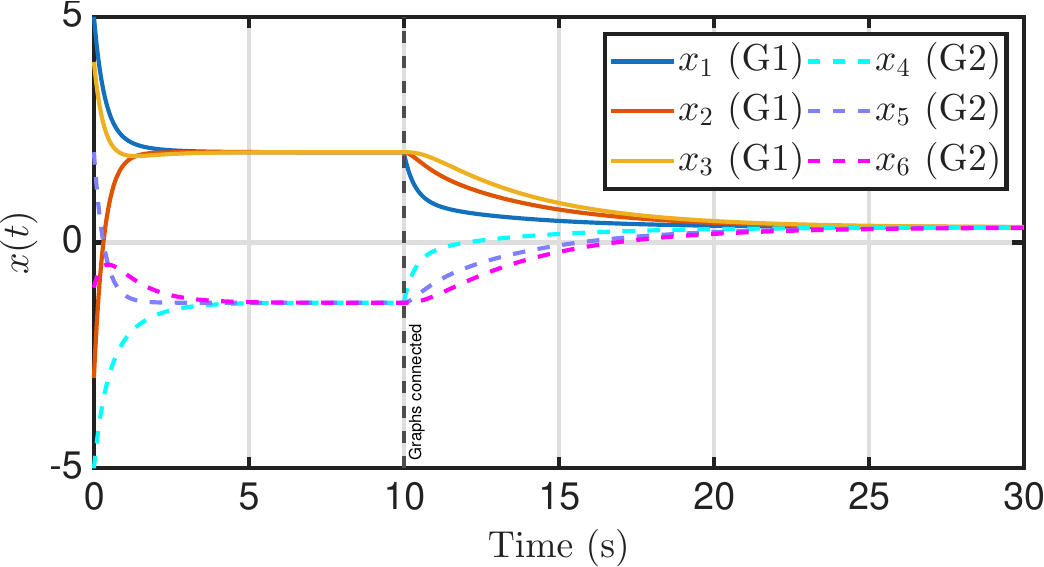}
\vspace{-5pt}
\caption{The effect of single bridge edge on the convergence rate of networked multi-agent systems.}\label{one_bridgeEX1}
\captionsetup{justification=centering}
\vspace{-5pt}
\end{center}
\end{figure}

It follows that the Fiedler eigenvalue, namely the second smallest eigenvalue of the Laplacian matrix associated with the combined graph, is $\lambda_2 (L_{12}) = 0.382$. This confirms the results of Proposition \ref{Prop:Bridge Gluing} that $\lambda_2 (L_{12}^1) \leq (1/3+1/3) \leq 0.67$. The consensus time of Graph 1, Graph 2 and the combined graph is shown in Fig. \ref{one_bridgeEX1}.

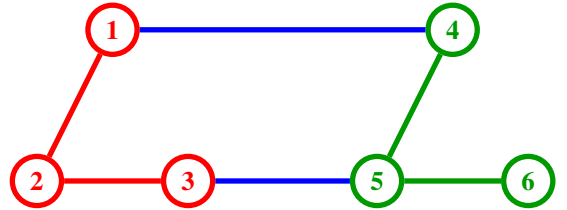
\begin{figure}[h]
\centering
\begin{tikzpicture}[
    every node/.style = {
        circle, draw,
        line width=2pt,
        minimum size = 0.65cm,
        font = \bfseries
    }
]

%% --- Red Triangle (1,2,3) ---
\node[draw=red, text=red] (1) at (-1, 2)    {1};
\node[draw=red, text=red] (2) at (-2, 0)   {2};
\node[draw=red, text=red] (3) at (0, 0)    {3};

\draw[red, line width=2pt] (1) -- (2);
\draw[red, line width=2pt] (2) -- (3);
%\draw[red, line width=2pt] (3) -- (1);

%% --- Green Triangle (4,5,6) ---
\node[draw=green!60!black, text=green!60!black] (4) at (3.5, 2)  {4};
\node[draw=green!60!black, text=green!60!black] (5) at (2.5, 0)  {5};
\node[draw=green!60!black, text=green!60!black] (6) at (4.5, 0)  {6};

\draw[green!60!black, line width=2pt] (4) -- (5);
\draw[green!60!black, line width=2pt] (5) -- (6);
%\draw[green!60!black, line width=2pt] (6) -- (4);

%% --- Red connecting line from 1 to 4 ---
\draw[blue, line width=2pt] (1) -- (4);
\draw[blue, line width=2pt] (3) -- (5);

\end{tikzpicture}
\caption{The graph gluing with two bridges.}
\label{fig:two-bridges}
\end{figure}

\begin{figure}[!ht]
\begin{center}
\includegraphics[scale=0.45]{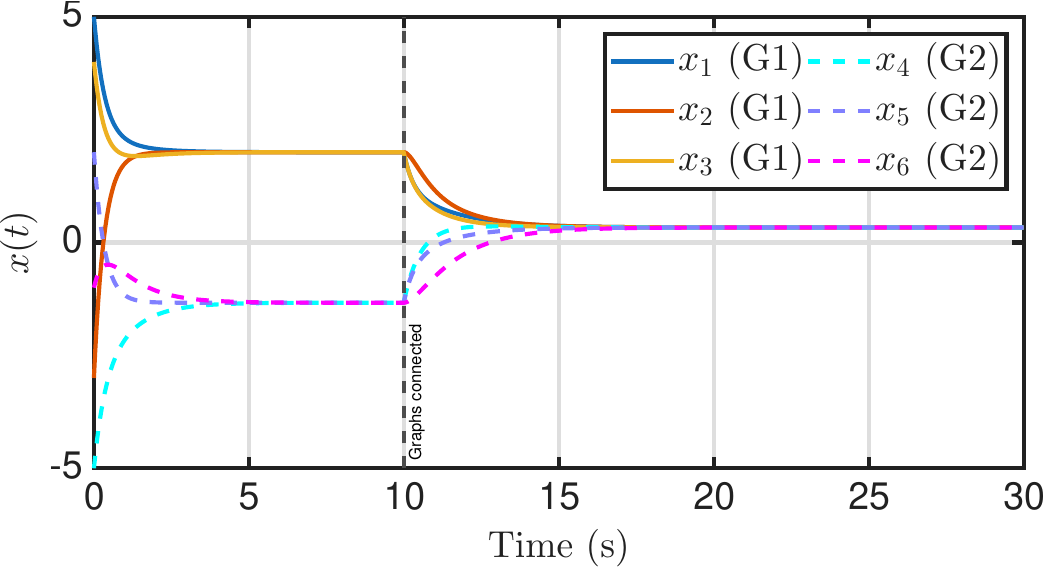}
\vspace{-5pt}
\caption{The Effect of two bridge edges on the convergence rate of graphs.}\label{twobridge_EX1}
\captionsetup{justification=centering}
\vspace{-5pt}
\end{center}
\end{figure}

Now, consider the case where two graphs are connected through two bridge edges, i.e. i.e. $k=2$ in \eqref{gluing} as shown in Fig. \ref{fig:two-bridges}. The corresponding Laplacian matrix becomes 
\begin{equation}
    L_{12}^2 = \left[ {\begin{array}{*{20}{c}}
  2&{ - 1}&0 &-1 &0 &0 \\ 
  { - 1}&2&{ - 1}&0 & 0&0\\ 
  0&{ - 1}&2 & 0& -1&0\\
  -1&0&0&2 &-1 &0\\
  0&0&-1 & -1&2 &0\\
  0&0&0 &0 & -1&1\\
\end{array}} \right]
\end{equation}

Now, the Fiedler eigenvalue of the Laplacian matrix associated with the combined graph, is $\lambda_2 (L_{12}^2) = 1$. This is consistent with the results of Proposition \ref{Prop:Bridge Gluing}, which yields the bound $\lambda_2 (L_{12}^2) \leq 2\times(1/3+1/3)= 1.34$. The consensus time of the multi-agent system communicating with Graph 1, Graph 2, and the combined graph are shown in Fig. \ref{twobridge_EX1}. As given in \cite{lewis2013cooperative}, in a multi-agent system with graph Laplacian matrix $L$, the consensus is reached with a time constant given by $\tau = 1/\lambda_2$, where $\lambda_2$ is the Fiedler eigenvalue. These results confirm that increasing the number of interconnecting edges between two graphs increases the Fiedler eigenvalue, thereby leading to faster convergence of the combined network.

\subsection{Example 2 (Bridge gluing in multi-agent systems with larger graphs)}
Now, consider two multi-agent systems with single integrator dynamics $\dot{x}_i = u_i$ and the distributed controller $u_i=\sum (x_i-x_j)$ are communicating over a larger graph topology depicted in Fig. \ref{64graphs_onebridge}. The initial conditions of agents in Graph 1 are $x^1(0)=[3,1,2,-1,2,-1]^T$ and for Graph 2 are $x^2(0)=[-2,2,-1,1]$.

\begin{figure}[h]
\centering
\scalebox{0.9}{  % <-- reduce this number to shrink more

\begin{tikzpicture}[
    every node/.style={circle, draw, line width=1.5pt,
        minimum size=0.6cm, font=\bfseries\small}
]

%% ---- Graph 1: Pentagon with center node (teal) ----
\node[fill=teal!70!blue, text=white] (1) at ($(90:1.5)  + (0,0)$)  {1};
\node[fill=teal!70!blue, text=white] (2) at ($(18:1.5)  + (0,0)$)  {2};
\node[fill=teal!70!blue, text=white] (3) at ($(306:1.5) + (0,0)$)  {3};
\node[fill=teal!70!blue, text=white] (4) at ($(234:1.5) + (0,0)$)  {4};
\node[fill=teal!70!blue, text=white] (5) at ($(162:1.5) + (0,0)$)  {5};
\node[fill=teal!70!blue, text=white] (6) at (0,0)                  {6};

\draw[teal!70!blue, line width=1.5pt]
    (1)--(2) (2)--(3) (3)--(4) (4)--(5) (5)--(1)
    (6)--(1) (6)--(3) (6)--(4);

%% ---- Graph 2: Sparse graph (gray) — shifted right ----
\node[fill=gray!25] (g1) at (3.5, 0.46)   {7};
\node[fill=gray!25] (g2) at (4.8, 1.2)   {8};
\node[fill=gray!25] (g3) at (6.1, 1.2)   {9};
\node[fill=gray!25] (g4) at (4.8,-0.2)   {10};

\draw[line width=1.5pt]
    (g1)--(g2) (g1)--(g4) (g2)--(g4) (g2)--(g3);

%% ---- Inter-graph edge: G1-node2 <--> G2-node1 (red) ----
\draw[red, line width=1.5pt] (2) -- (g1);

\end{tikzpicture}

} % end scalebox
\caption{Two connected graphs: teal (G1) and gray (G2) with inter-graph edge in red.}
\label{64graphs_onebridge}
\end{figure}

\begin{figure}[!ht]
\begin{center}
\includegraphics[scale=0.45]{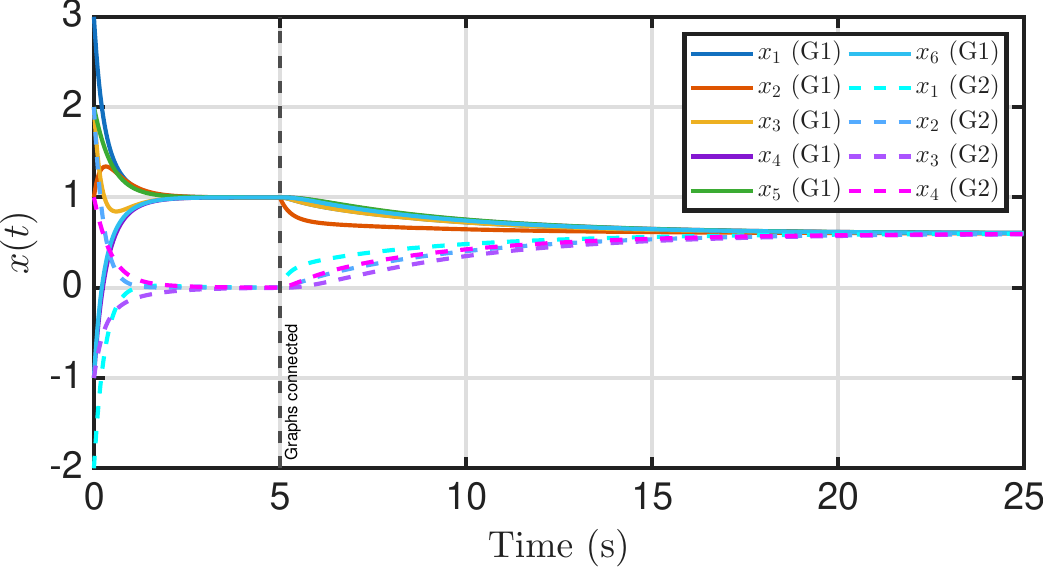}
\vspace{-5pt}
\caption{The effect of two bridge edges on the convergence rate of graphs.}\label{64_graph_onebridge}
\captionsetup{justification=centering}
\vspace{-5pt}
\end{center}
\end{figure}

The graph Laplacian matrix associated to each graph is
\begin{equation*}
    L_1 = \left[ {\begin{array}{*{20}{c}}
  3&{ - 1}&0 &0 &-1 &-1 \\ 
  { - 1}&2&{ - 1}&0 & 0&0\\ 
  0&{ - 1}&3 & -1& 0&-1\\
  0&0&-1&3 &-1 &-1\\
  -1&0&0 & -1&2 &0\\
  -1&0&-1 &-1 & 0&3\\
\end{array}} \right]
\end{equation*}

\begin{equation*}
    L_2 = \left[ {\begin{array}{*{20}{c}}
  2&{ - 1}&0 &-1 \\ 
  -1&3&-1&-1 \\ 
  0&{ - 1}&1 & 0\\
  -1&-1&0&2 \\
\end{array}} \right]
\end{equation*}

The Laplacian matrix of the combined graph with one bridge as shown in Fig. \ref{64graphs_onebridge} becomes
\[
L_{12}^1 =
\begin{bmatrix}
 3 & -1 &  0 &  0 & -1 &  -1 &  0 &  0 &  0 &  0 \\
-1 &  3 & -1 &  0 &  0 &  0 & -1 &  0 &  0 &  0 \\
 0 & -1 &  3 & -1 &  0 & -1 &  0 &  0 &  0 &  0 \\
 0 &  0 & -1 &  3 & -1 & -1 &  0 &  0 &  0 &  0 \\
-1 &  0 &  0 & -1 &  2 &  0 &  0 &  0 &  0 &  0 \\
-1 &  0 & -1 & -1 &  0 &  3 &  0 &  0 &  0 &  0 \\
 0 & -1 &  0 &  0 &  0 &  0 &  3 & -1 &  0 & -1 \\
 0 &  0 &  0 &  0 &  0 &  0 & -1 &  3 & -1 & -1 \\
 0 &  0 &  0 &  0 &  0 &  0 &  0 & -1 &  1 &  0 \\
 0 &  0 &  0 &  0 &  0 &  0 & -1 & -1 &  0 &  2
\end{bmatrix}
\]

Then, the Fiedler eigenvalue of the Laplacian matrix associated with the combined graph, is $\lambda_2 (L_{12}^1) = 0.2208$. This confirms the results of Proposition \ref{Prop:Bridge Gluing} that $\lambda_2 (L_{12}^1) \leq (1/6+1/4)= 0.4167$. For the combined graph, the corresponding convergence time constant is approximately $\tau = 1/\lambda_2=4.53$ seconds. The consensus time of multi-agent system with the communication topology given in Graph 1, Graph 2 and the combined graph is illustrated in Fig. \ref{64_graph_onebridge}.

\begin{figure}[h]
\centering
\scalebox{0.9}{  % <-- reduce this number to shrink more

\begin{tikzpicture}[
    every node/.style={circle, draw, line width=1.5pt,
        minimum size=0.6cm, font=\bfseries\small}
]

%% ---- Graph 1: Pentagon with center node (teal) ----
\node[fill=teal!70!blue, text=white] (1) at ($(90:1.5)  + (0,0)$)  {1};
\node[fill=teal!70!blue, text=white] (2) at ($(18:1.5)  + (0,0)$)  {2};
\node[fill=teal!70!blue, text=white] (3) at ($(306:1.5) + (0,0)$)  {3};
\node[fill=teal!70!blue, text=white] (4) at ($(234:1.5) + (0,0)$)  {4};
\node[fill=teal!70!blue, text=white] (5) at ($(162:1.5) + (0,0)$)  {5};
\node[fill=teal!70!blue, text=white] (6) at (0,0)                  {6};

\draw[teal!70!blue, line width=1.5pt]
    (1)--(2) (2)--(3) (3)--(4) (4)--(5) (5)--(1)
    (6)--(1) (6)--(3) (6)--(4);

%% ---- Graph 2: Sparse graph (gray) — shifted right ----
\node[fill=gray!25] (g1) at (3.5, 0.46)   {7};
\node[fill=gray!25] (g2) at (4.8, 1.2)   {8};
\node[fill=gray!25] (g3) at (6.1, 1.2)   {9};
\node[fill=gray!25] (g4) at (4.8,-0.2)   {10};

\draw[line width=1.5pt]
    (g1)--(g2) (g1)--(g4) (g2)--(g4) (g2)--(g3);

%% ---- Inter-graph edge: G1-node2 <--> G2-node1 (red) ----
\draw[red, line width=1.5pt] (2) -- (g1) (1) -- (g2) (3) -- (g4);

\end{tikzpicture}

} % end scalebox
\caption{Two connected graphs: teal (G1) and gray (G2) with inter-graph edge in red.}
\label{64_3bridges}
\end{figure}

Now, consider the case where two graphs are connected through two bridge edges, i.e. i.e. $k=3$ in \eqref{gluing} as shown in Fig. \ref{64_3bridges}. The corresponding Laplacian matrix becomes 
\[
L_{12}^3 =
\begin{bmatrix}
 4 & -1 &  0 &  0 & -1 &  -1 &  0 &  -1 &  0 &  0 \\
-1 &  3 & -1 &  0 &  0 &  0 & -1 &  0 &  0 &  0 \\
 0 & -1 &  4 & -1 &  0 & -1 &  0 &  0 &  0 &  -1 \\
 0 &  0 & -1 &  3 & -1 & -1 &  0 &  0 &  0 &  0 \\
-1 &  0 &  0 & -1 &  2 &  0 &  0 &  0 &  0 &  0 \\
-1 &  0 & -1 & -1 &  0 &  3 &  0 &  0 &  0 &  0 \\
 0 & -1 &  0 &  0 &  0 &  0 &  3 & -1 &  0 & -1 \\
 -1 &  0 &  0 &  0 &  0 &  0 & -1 &  4 & -1 & -1 \\
 0 &  0 &  0 &  0 &  0 &  0 &  0 & -1 &  1 &  0 \\
 0 &  0 &  -1 &  0 &  0 &  0 & -1 & -1 &  0 &  3
\end{bmatrix}
\]

\begin{figure}[!ht]
\begin{center}
\includegraphics[scale=0.45]{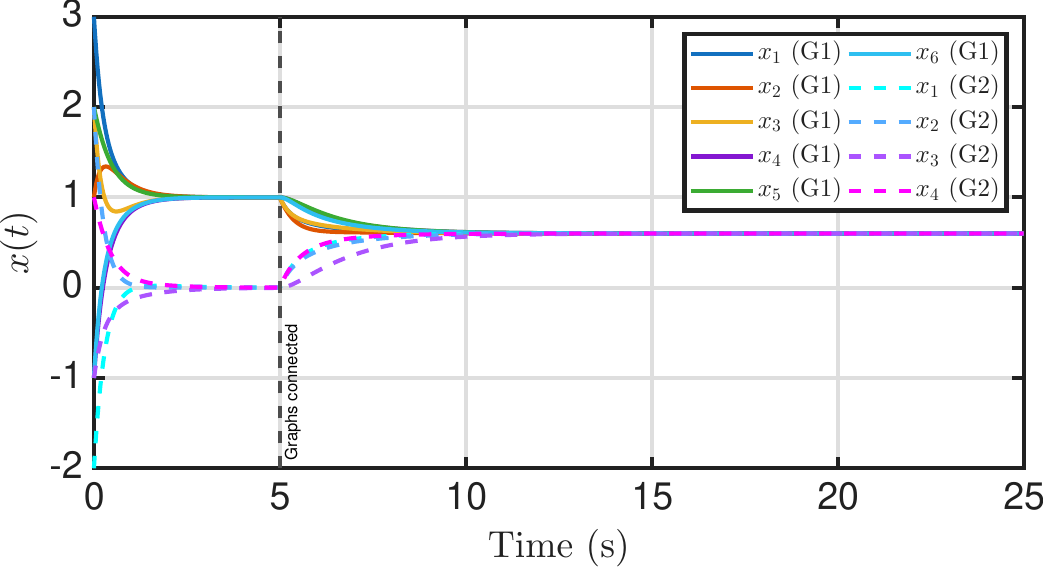}
\vspace{-5pt}
\caption{The Effect of two bridge edges on the convergence rate of graphs.}\label{64_graph_3bridge}
\captionsetup{justification=centering}
\vspace{-5pt}
\end{center}
\end{figure}

Now, the Fiedler eigenvalue of the Laplacian matrix associated with the combined graph, is $\lambda_2 (L_{12}^3) = 0.6614$. This is consistent with the results of Proposition \ref{Prop:Bridge Gluing}, which yields the bound $\lambda_2 (L_{12}^3) \leq 3\times(1/6+1/4)= 1.25$. The consensus time of the multi-agent system communicating with Graph 1, Graph 2, and the combined graph are shown in Fig. \ref{64_graph_3bridge}. It can be seen that the graph with three bridges has a time constant of $\tau=1/0.664=1.5$ seconds which shows faster convergence. These results confirm that increasing the number of interconnecting edges between two graphs increases the Fiedler eigenvalue, thereby leading to faster convergence of the combined network.

\section{CONCLUSION}
In this paper, the effects of graph gluing operations in networks of multi-agent systems and their influence on consensus performance were investigated. In particular, the study examined how interconnecting multiple multi-agent subsystems through communication links modifies the underlying graph topology and alters the spectral properties of the resulting network. Both bridge gluing and interface gluing were analyzed to understand their roles in shaping the connectivity of the combined graph. The analysis focused on how the structure and number of communication links between subsystems affect the Fiedler eigenvalue of the graph Laplacian, which is a key parameter governing the convergence rate and consensus time of the networked system. The results demonstrate that increasing the number of bridging links enhances algebraic connectivity. Overall, the findings establish a clear relationship between interconnection strategies, algebraic connectivity, and consensus performance. The proposed framework offers useful insight for the design and analysis of interconnected multi-agent systems, where appropriate selection of gluing structures can significantly improve convergence behavior. Simulation results were presented to validate the theoretical analysis and to illustrate the impact of different gluing mechanisms on network performance.

\section{ACKNOWLEDGMENTS}
Research was sponsored by the Office of Naval Research and was accomplished under Grant Number W911NF-24-1-0293. The views and conclusions contained in this document are those of the authors and should not be interpreted as representing the official policies, either expressed or implied, of the Office of Naval Research or the U.S. Government. The U.S. Government is authorized to reproduce and distribute reprints for Government purposes notwithstanding any copyright notation herein.
\bibliographystyle{ieeetr}

\bibliography{reference}

%\addtolength{\textheight}{-3cm}

\begin{comment}
Therefore, parameter $a$ must be selected such that
\begin{equation}\label{eq:A12}
    a\leqslant 2 \lambda_{min}\left(L_{nr \times nr}\right)
\end{equation}
Now we investigate the stability of the system in two cases:

\emph{Case1.} In this case we suppose that there is no attack on agents, i.e. $f_{nr}=0$. From \eqref{eq:A9}, the triggering parameter $\eta_i$ needs to be designed using the following condition
\begin{equation}\label{eq:A13}
    \eta_{i,max} < \sqrt{\frac{2a\lambda_{min}\left(L_{nr \times nr}\right))-a^2}{\lambda_{max}\left( {L_{nr \times nr}}^T L_{nr \times nr}\right)}}
\end{equation}
Enforcing condition \eqref{eq:A13}, one has $\alpha > \beta$ and hence $\dot V \leq 0$.
\end{comment}
\end{document}